\documentclass[letterpaper, 10 pt, conference]{ieeeconf}  

\IEEEoverridecommandlockouts                              

\usepackage{graphicx}      
\usepackage{amsmath, amsfonts}
\usepackage{cite}
\usepackage{color}
\usepackage{multirow}
\usepackage{mathtools}

\usepackage{tikz}
\usetikzlibrary{calc,positioning}
\tikzstyle{solid node}=[circle,draw,inner sep=.75,fill=black]

\newcommand{\R}{\ensuremath{{\mathbb R}}}

\newtheorem{thm}{\bf Theorem}
\newtheorem{assumption}{\bf Assumption}
\newtheorem{lemma}{\bf Lemma}

\newtheorem{remark}{\bf Remark}

\newcommand{\ifnote}[1]{} 

\title{\LARGE \bf
An Output Feedback Stabilizer for MIMO Nonlinear Systems with Uncertain Input Gain: Nonlinear Nominal Model
}

\author{Wonseok Ha$^{1}$ and Juhoon Back$^{1}$
\thanks{*This work was not supported by any organization}
\thanks{$^{1}$W. Ha and J. Back are with the School of robotics, Kwangwoon University, Seoul, Korea. Emails: {\tt\small wonseok00@kw.ac.kr, backhoon@kw.ac.kr}}%
\thanks{Corresponding author: J. Back}%
}

\begin{document}

\maketitle
\thispagestyle{empty}
\pagestyle{empty}

\begin{abstract}
This paper deals with the output feedback stabilization problem of nonlinear multi-input multi-output systems having an uncertain input gain matrix. 
It is assumed that the system has a well-defined vector relative degree and that the zero dynamics is input-to-state stable. Based on the assumption that there exists a state feedback controller which globally asymptotically stabilizes the origin of the nominal closed-loop system, we present an output feedback stabilizer which recovers the stability of the nominal closed-loop system in the semi-global practical sense. Compared to previous results, we allow that the nominal system can have a nonlinear input gain matrix that is a function of state and this is done by modifying the structure of the disturbance observer-based robust output feedback controller. It is expected that the proposed controller can be well applied to the case when the system's nonlinearity is to be exploited rather than canceled.
\end{abstract}
\section{Introduction}
In most practical control systems, only parts of system states can be measured. To control such systems, one can simply construct a static output feedback controller such as proportional controller, or design a dynamic output feedback controller such as proportional-integral controller. When the system model is available, we can construct a state observer to estimate the system state and use the state feedback controller with the estimated state. The latter idea is one of standard solutions for linear systems and the separation principle guarantees that the idea actually works \cite{Kailath1980}.


The idea of combining a state feedback controller and a state estimator does not work well for general nonlinear systems and actually it is quite challenging as explained in books and papers including \cite{B_Khalil,bIsidori,Teel94,AK1999}. Fortunately, it turns out that the separation principle holds true if we consider a (possibly very large) compact region of attraction rather than the whole state space and the basic idea is to employ a high-gain observer whose convergence speed can be tuned arbitrarily fast and devise a structure that prevents possibly very large overshoot of the observer from propagating to the system states; see the works, e.g., \cite{Khalil93, Teel94, Teel94SIAM, Lin1995TAC,FK2008} for details. 
  

The aim of this paper is to develop an output feedback stabilizing controller for multi-input multi-output (MIMO) nonlinear systems which contain uncertainties in the input gain matrix. Although uncertainties in the input gain have been considered in many studies including \cite{Nussbaum1983,KKK95,Xu2003TAC}, most of them assume that state variables of the system are available for feedback or considered single input systems. In MIMO systems, the uncertainty in the input gain makes the control problem very hard because the actual direction of an input applied to the system as well as the magnitude can be different from what is expected while in the single input case the direction is fixed. When we consider the output feedback case, the problem becomes obviously more difficult.

Some relevant results on the problem are available in the literature. In \cite{Ge2014TAC}, an adaptive output feedback controller is presented for MIMO nonlinear systems with input uncertainties. It is assumed that the system has an identical relative degree between inputs and outputs, and the input uncertainty is composed of  an uncertain diagonal term that depends on the states and an unknown constant matrix. More recently, a robust output feedback stabilizer has been proposed for a partially feedback linearizable system having a vector relative degree and an uncertain input gain matrix in \cite{Wang2015}. The uncertainty is assumed to be close to a constant known matrix (called `guess of input gain') in some sense and an extended high gain observer is employed to estimate the high derivatives of outputs so that the effect of uncertainties can be estimated. A disturbance observer-based solution has been proposed in \cite{BackShim2009}, where the uncertain input gain matrix is assumed to belong to a known sector and two filters with sufficiently high bandwidth work together to estimate the effect of uncertainties and external disturbances; see \cite{Ohnishi87} for the idea of a disturbance observer, \cite{Chen+2016TIE} for a recent survey, and \cite{Shim+2016CTT} for a tutorial. It is shown that the real closed-loop system behaves like a nominal linear system for whole time horizon. As in the work \cite{Wang2015}, a constant gain matrix is introduced in \cite{BackShim2009} which can be regarded as a nominal gain or a target gain for the nominal closed-loop system. 

In this paper, we would like to consider more general case where the nominal system has a nonlinear nominal gain matrix that depends on the system state rather than a constant one. In many systems such as robot manipulators and aerial vehicles, the input gain matrix depends on the posture or attitude of the system and thus the input gain is naturally nonlinear. If some uncertainties in system parameters such as length and mass are present, which is often the case, the input gain matrix becomes uncertain. In this case, if we apply control strategies that rely on a linear nominal model, then the discrepancy between the actual system and the nominal one will be considerably large, resulting unnecessarily large control effort to compensate for the effect of uncertainties. Thus, our motivation is to consider more realistic assumption on the systems and allow the compensated system to have possibly good intrinsic nonlinearities. 

The contribution of this paper can be summarized as follows. Firstly, we present a new structure of the disturbance observer-based controller for MIMO nonlinear systems that have uncertainties in the input gain matrix whose nominal counterpart can be a nonlinear function. Secondly, a new type of the uncertainty on the input gain matrix involving the nonlinear nominal model is presented. Finally, it is shown that the separation principle holds in our case in the sense that the state feedback controller for the given nominal system and the observer or the filter that estimates system outputs and the effect of uncertainties can be designed separately. 

%
%
%

The rest of this paper is organized as follows. The output feedback stabilization problem for MIMO nonlinear systems is formulated in Section \ref{sec:problem_formulation}. A robust output feedback controller is proposed and the stability of the closed-loop system is analyzed in Section \ref{sec:main_result}. An example to validate the proposed method is given in Section \ref{sec:example}. Finally, conclusions and outlooks are presented in Section \ref{sec:conclusion}.

{\bf Notation}:
For two vectors $x\in \mathbb R^{n_1}$ and $y\in \mathbb R^{n_2}$, the concatenated vector $[x;y]$ is defined by $[x;y]=\begin{bmatrix}x^\top & y^\top \end{bmatrix}^\top$.
Concatenation of multiple vectors (or scalars) are defined similarly.
Given $n$ vectors $x_1, \dots, x_n$ in $\mathbb R^\nu$, $x_{[i]}\in \mathbb R^n$ is defined by $x_{[i]}=\begin{bmatrix} x_{1i}; x_{2i}; \cdots; x_{ni}\end{bmatrix}$. For two sets $\Omega_x$ and $\Omega_y$, the Cartesian product of sets $\Omega_x \times \Omega_y$ is denoted by $\Omega_{(x,y)}$. For a $k$-times  differentiable function $x(t)$, $x^{(k)}(t)$ denotes $\frac{d^k x}{dt^k}(t)$. The vector $0_k$ stands for the zero vector in $\R^k$.

\section{Problem  Formulation}	\label{sec:problem_formulation}
We consider a nonlinear system which has $m$ inputs and $m$ outputs and admits a well-defined vector relative degree $\{ \nu_1, \nu_2, \dots, \nu_m \}$. This system can be written in the Byrnes-Isidori normal form \cite{bIsidori}  
\begin{equation}	\label{eq:system}
\begin{split}
\dot z &= F_0(z,x)\\
\dot x &= Ax + B(F(z,x) + G(z,x,t)u)\\
y &= Cx
\end{split}
\end{equation}
where $u\in \R^m$ is the control input, $y \in \R^m$ is the system output, and $x \in \R^\nu$ and $z \in \R^{n-\nu}$ are system states with $\nu = \nu_1+\cdots+\nu_m$.  
The matrices $A$, $B$, and $C$ are given by
\begin{align*}
A &= {\rm diag}\{A_1, \dots, A_m\}\\
B &= {\rm diag}\{B_1, \dots, B_m\}\\
C &= {\rm diag}\{C_1, \dots, C_m\}
\end{align*}
where
\begin{align*}
A_i &:= \begin{bmatrix} 0_{\nu_i-1} & I_{\nu_i-1} \\ 0 & 0_{\nu_i-1}^{\top} \end{bmatrix},  B_i := \begin{bmatrix} 0_{\nu_i-1} \\ 1\end{bmatrix}, 
C_i := \begin{bmatrix} 1 & 0_{\nu_i-1}^{\top} \end{bmatrix}.
\end{align*}
For convenience, we decompose the state vector $x$ as $x = \begin{bmatrix}x_1;x_2;\cdots;x_m\end{bmatrix}$ with $x_i\in \R^{\nu_i}$ being given by
$x_i = \begin{bmatrix}x_{i1};x_{i2};\dots;x_{i\nu_i}\end{bmatrix}$.

It is assumed that the functions $F_0(z,x)$ and $F(z,x)$ are smooth. We assume that, for any $(z,x,t)\in\R^{n+1}$, the input gain matrix $G(z,x,t)$ is smooth and invertible. It is noted that the uncertainty in $G$ is the main concern of this paper and we do not consider uncertainties in $F_0$ or $F$ for the sake of simplicity. 

We introduce a nominal gain matrix $\bar G(z,x)$ which is a nominal function for $G(z,x,t)$ (hence known) and assume that it is smooth for any $(z,x)\in\R^n$. With $\bar G$, a nominal closed-loop system is defined by
\begin{equation}	\label{eq:nominal}
\begin{split}
\dot {\bar z} &= F_0({\bar z}, {\bar x})\\
\dot {\bar x} &= A\bar x + B(F(\bar z,\bar x)+\bar G(\bar z,\bar x)u_r)
\end{split}
\end{equation}
where $u_r$ is the control input that has been designed for the nominal system. 
\begin{assumption}	\label{ass:ISS}
The system $\dot z = F_0(z,u_z)$ is input-to-state stable with respect to the input $u_z$ \cite{Sontag95}. Precisely, for any initial state $z(0) \in \mathbb R^{n-\nu}$ and any piecewise continuous bounded input $u_z$, the solution $z(t)$ exists for all $t \ge 0$ and there exists a class $\mathcal {KL}$ function $\beta_z$ and a class $\mathcal K$ function $\gamma_z$ such that  
\begin{align*}
 \quad\| z(t) \| \le \beta_z(\| z(0) \|, t) + \gamma_z \left( \sup_{0\le s \le t} \|u_z(s)\| \right), \forall t\ge 0. \ \diamond
\end{align*}
\end{assumption}

\begin{assumption}	\label{ass:nominal_control_input}
The origin of the nominal system \eqref{eq:nominal} can be made globally asymptotically stable by a state feedback control of the form 
$u_r= U_r(\bar z, \bar x)$. $\hfill$$\diamond$

By the converse Lyapunov theorem \cite[Theorem 4.17]{B_Khalil}, Assumption \ref{ass:nominal_control_input} ensures that there exists a smooth, positive definite function $\bar V(\bar z, \bar x)$ such that 
\begin{equation}		\label{eq:Lyapunov_theory}
\begin{split}
&\alpha_1 \left(\left\|[ \bar z ; \bar x ]\right\| \right)\le \bar V(\bar z,\bar x) \le \alpha_2 \left(\left\|[ \bar z ; \bar x ]\right\| \right) \\
&\begin{multlined} \frac{\partial \bar V}{\partial \bar x}[ A\bar x + B(F(\bar z, \bar x)+\bar G(\bar z, \bar x) U_r(\bar z, \bar x)] \\ + \frac{\partial \bar V}{\partial \bar z} F_0(\bar z, \bar x) \!\le\! -\alpha_3(\|\bar z;\bar x\|) \end{multlined}
\end{split}
\end{equation}
where $\alpha_1$ and $\alpha_2$ are class $\mathcal K_\infty$ functions, and $\alpha_3$ is a class $\mathcal K$ function.
\end{assumption}



The main concern of this paper is to deal with the uncertainty in the input gain matrix $G(z,x,t)$. We allow that the nominal input gain matrix, denoted by $\bar G(z,x)$, can be a nonlinear function of state. This can be regarded as a generalization of the assumption introduced in \cite{Wang2015}, saying that there exists a constant matrix ${\bf G}$ such that $
\max_{\Lambda \, {\rm diagonal}, |\Lambda| \le 1} |[G(z,x)-{\bf G}]\Lambda {\bf G}^{-1}| < 1, \forall (z,x) \in \R^n$.
 
\begin{assumption}	\label{ass:contraction}
The uncertain input gain matrix $G(z,x,t)$ and its nominal one $\bar G(z,x)$ are close in the sense that 
\begin{align*}
\|I- G(z,x,t)\bar G^{-1}(z,x)\| \le \mu < 1, \forall (z,x)\in \R^n,  \forall t.\quad \diamond
\end{align*}
\end{assumption}
\medskip

The objective of this paper is to design a dynamic output-feedback  controller 
which makes the origin of the closed-loop system 
semi-globally practically stable despite the uncertainty in the input gain matrix.
By the semi-global practical stability, we mean that for any given initial condition sets $S_z$ and $S_x$ of $z$ and $x$, respectively, we can design the controller so that the trajectory of the closed-loop system is uniformly ultimately bounded and the size of the ultimate bound can be tuned arbitrarily. 

\section{Main Result}	\label{sec:main_result}
In this section, we firstly propose a robust output feedback stabilizer with a design procedure and then analyze the stability of the closed-loop system.

\subsection{Proposed Controller}\label{sec:Proposed_Controller}
We present a controller given by
\begin{equation}	\label{eq:proposed_controller}
\begin{split}
\dot{\bar z} &= F_0(\bar z, \phi(q)) \\
\dot q&= A_{a\tau}q + B_{a\tau}^q y\\
\dot p &= A_{a\tau}p + B_{a\tau}^p \bar G(\bar z, \phi(q)) u\\
w &= Cp+ B^\top B_{a\tau}^q (Cq-y)+F(\bar z,\phi(q)) \\
u &= \bar G^{-1}(\bar z,\phi(q))\Phi(w)+U_r(\bar z, \phi(q))
\end{split}
\end{equation}
where $\bar z \in \mathbb R^{n-\nu}$, $q\in \mathbb R^\nu$, and $p\in \mathbb R^\nu$ are internal states of the controller. Noting that $q$ and $p$ have the same dimension as $x$, we decompose $q$ and $p$ like $x$, e.g., $q=\begin{bmatrix}q_1;\cdots; q_m\end{bmatrix}$ with $q_i=\begin{bmatrix}q_{i1};\cdots; q_{i\nu_i}\end{bmatrix}$. 

The controller involves several matrices ($A_{a\tau}$, $B_{a\tau}^q$, $B_{a\tau}^p$) and functions ($\Phi$, $\phi$) to be determined and in what follows we explain how to construct them.

Firstly, we explain the dynamics of $q$ and $p$. They share the same system matrix given by $A_{a\tau}= {\rm diag}\{A_{a_1\tau}, \dots, A_{a_m\tau}\}$ where
\begin{align*}
A_{a_i\tau} &= \begin{bmatrix} \begin{array}{c} -\frac{a_{i\nu_i}}{\tau}\\ \vdots \end{array}  & I_{\nu_i-1} \\  -\frac{a_{i1}}{\tau^{\nu_i}} &  0_{\nu_i-1}^{\top}  \end{bmatrix} \in \R^{\nu_i \times \nu_i}.
\end{align*}
The input matrices $B_{a\tau}^q$ and $B_{a\tau}^p$ are given by 
$B_{a\tau}^q= {\rm diag}\{B_{a_1\tau}^q, \dots, B_{a_m\tau}^q\}$ and $B_{a\tau}^p= {\rm diag}\{B_{a_1\tau}^p, \dots, B_{a_m\tau}^p\}$  where
\begin{align*}
B_{a_i\tau}^q = \begin{bmatrix}\frac{a_{i\nu_i}}{\tau};\cdots;\frac{a_{i1}}{\tau^{\nu_i}}\end{bmatrix}, \quad B_{a_i\tau}^p = \begin{bmatrix}0;\cdots;0;\frac{a_{i1}}{\tau^{\nu_i}}\end{bmatrix}.
\end{align*}
The constants $a_{i1},\dots, a_{i\nu_i}$, and $\tau$ in the matrices are design parameters to be chosen later. 

\begin{remark}Note that the dynamics of $q$ can be decomposed into  $\dot q_i = A_{a_i\tau}q_i + B_{a_i \tau}^q y_i $ and this is the well known high-gain observer which estimates $y_i$ and its time derivatives. In addition, the dynamics of $p_i$ is obtained as $\dot p_i = A_{a_i\tau}p_i + B_{a_i\tau}^p \bar G_i(\bar z, \phi(q)) u$ where $\bar G_i(\bar z, \phi(q))$ is the $i$th row vector of $\bar G_i(\bar z, \phi(q))$, and this subsystem takes $\bar G_i(\bar z, \phi(q)) u$ as the input, and the signal $p_{i1}=C_ip_i$ (can be regarded as the output of $p_i$-dynamics) together with the signal $B_i^\top B_{a_i\tau}^q(C_i q_i -y_i)$ which is nothing but, $\dot q_{i\nu_i}$ is used to generate a feedback signal $w_i$. This signal can be regarded as an estimate of the effect of the discrepancy between $G$ and $\bar G$; see for example \cite{BackShim2009}.   Fig. \ref{fig:OFDOB} shows the structure of the proposed controller and the flow of signals. Here, $P$ denotes the system \eqref{eq:system}. \hfill $\diamond$
\end{remark}

Now, we provide the details on  the design of parameters $a_{i1},\dots, a_{i\nu_i}$. For simplicity, let $a_i = \begin{bmatrix} a_{i1} ; a_{i2} ; \cdots ; a_{i\nu_i} \end{bmatrix}$, $i=1,\dots,m$. For each $i$, choose $a_{i2},\dots, a_{i\nu_i}$ such that
\begin{equation}	\label{eq:a_i}
s^{\nu_i-1} + a_{i\nu_i} s^{\nu_i-2} + \cdots + a_{i2} = 0
\end{equation}
has all roots in the left half complex plane. With  $a_{i2}, \dots, a_{i\nu_i}$ fixed, we choose $a_{i1}$ as follows.
Define $D(1-\mu ,1+\mu)$ by the closed disk in the complex plane whose center and diameter are the midpoint of $-1/(1-\mu)+j0$ and $-1/(1+\mu)+j0$ and the distance between them, respectively. 
Let
\begin{equation}	\label{eq:tf}
H_i(s) = \frac{1}{s}\frac{a_{i1}}{s^{\nu_i-1} + a_{i\nu_i} s^{\nu_i-2} + \cdots + a_{i2}}
\end{equation}
and find a positive constant $a_{i1}$ such that the Nyquist plot of $H_i(s)$ lies strictly outside of the disk $D(1-\mu,1+\mu)$ and does not surround the disk.
Note that such $a_{i1}$ always exists since the Nyquist plot of $H_i(s)$ is bounded to the left half plane and sufficiently small $a_{i1}$ will reduce the size of the Nyquist plot so that the disk $D(1-\mu,1+\mu)$ does not touch the Nyquist plot, where $\mu$ is given in Assumption \ref{ass:contraction}. From the Nyquist stability theory, $a_{i1}$ guarantees that all the roots of $s^{\nu_i} + a_{i\nu_i} s^{\nu_i-1} + \dots + a_{i2}s+a_{i1} = 0$ belong to the left half complex plane.

We are now ready to determine the functions $\phi$ and $\Phi$. They are chosen as continuously differentiable saturation functions and the saturation levels are found considering the dynamics of the nominal system, the sizes of initial condition sets, and the Lyapunov function associated to the nominal closed-loop system. 

To proceed, let  $\delta_{\bar z}>0$, $\delta_{z}>0$, and $\delta_{x}>0$ be given and define $S_{\bar z}= \{\bar z  \in \R^{n-\nu} | \|\bar z\| \le \delta_{\bar z} \}$, $S_{z}= \{z  \in \R^{n-\nu} | \|z\| \le \delta_z \}$, and $S_x= \{x  \in \R^\nu | \|x\| \le \delta_x \}$. For a given constant $\delta>0$, define the level set $\Omega_{(\bar z, x)}(\delta)=\{[\bar z; x] \in \R^n \big| \bar V(\bar z, x) \le \delta\}$ where $\bar V$ is the Lyapunov function from \eqref{eq:Lyapunov_theory}.

Suppose the initial conditions $\bar z(0)$, $z(0)$ and $x(0)$ belong to  $S_{\bar z}$, $S_{z}$, $S_{x}$, respectively. Let  $l>0$ be any constant such that the level set $\Omega_{(\bar z, x)}(l)=\{[\bar z; x] \in \R^n \big| \bar V(\bar z, x) \le l\}$ strictly contains  $S_{\bar z}\times S_{x}$ in a way that the boundaries of $\Omega_{(\bar z,x)}(l)$ and $S_{\bar z}\times S_{x}$ are disjoint. 

From Assumption \ref{ass:ISS}, we have $\|z(t)\| \le \beta_z (\delta_z,0) + \gamma_z(\alpha_1^{-1}(l))=:l_z, \forall t\ge0$ provided that $(\bar z(t);x(t)) \in \Omega_{\bar z,x}(l)$, where $\alpha_1$ comes from \eqref{eq:Lyapunov_theory}.
Let $\delta_1$ and $\delta_w$ be positive constants. Define $\bar l=\alpha_2(\alpha_1^{-1}(l)+\delta_1)$ and let $L_{F}$ be the Lipschitz constant of $F(\bar z, x)$ over $\Omega_{(\bar z, x)}(\bar l)$.  With $\Omega_{(\bar z,x)}(\cdot)$ and the bound of $ z(t)$, we define differentiable saturation functions $\phi$ and $\Phi$ by
\begin{equation}
\begin{split}
(\bar z,\phi(x)) &= (\bar z,x), \quad \forall (\bar z,x) \in \Omega_{(\bar z, x)}(\bar l)\\
\Phi(w) &= w,\quad \forall w\in \Omega_w \\
\left\|\frac{\partial \phi}{\partial x}(x)\right\| &\le 1, \forall x\in \R^\nu\ {\rm and} \ \left\|\frac{\partial \Phi}{\partial w}(w)\right\| \le 1, \forall w\in \R^m  \!\!
\end{split}
\end{equation}
where 
\begin{align*}
\Omega_w&=  \Big\{ w\in \R^m| w(\bar z, z, x,\tilde w)=\\
& \quad\quad \bar G(\bar z, x) G^{-1}(z,x,t)\big[ F(\bar z, x)-F(z,x) \\&\hspace{2cm} + (\bar G(\bar z, x)-G(z,x,t)) U_r(\bar z, x)\big]+\tilde\omega, \\&\quad \quad   \|  z\|\in l_{  z},  (\bar z, x)\in \Omega_{(\bar z, x)}(l), \\ & \hspace{2cm} {\rm all\ possible\ }G,  \|\tilde\omega\|\le \delta_w  +L_F \delta_1 \Big\}.
\end{align*}

One can find the saturation level for $\Phi$ by computing the maximum of $w(\bar z, z,x)$ over all possible $(\bar z, z,x)$ and $G$ under consideration. In fact, any bounded saturation level which can cover the set $\Omega_w$ defined above  will do the job.

Based on the discussion so far, we state the main stability result of this paper. 

%
\begin{thm}	\label{thm:main_result}
Consider the closed-loop system \eqref{eq:system} and  \eqref{eq:proposed_controller}. Let $S_p$ and $S_q$ be the compact sets containing the origin and assume that $p(0) \in S_p$ and $q(0) \in S_q$. Let $l>0$ be such that $S_{\bar z} \times S_x \subset \Omega_{(\bar z, x)}(l)$ and the boundaries of $S_{\bar z} \times S_x$ and $\Omega_{(\bar z,x)}(l)$ are disjoint. Under Assumptions \ref{ass:ISS}-\ref{ass:contraction}, the controller \eqref{eq:proposed_controller} 
whose design parameters are chosen from the procedure described in this section ensures the following. For a given $\epsilon>0$, there exists $\tau^*>0$ and $T_\epsilon>0$ such that, for each $0<\tau<\tau^*$, the solution of the closed-loop system denoted by $[\bar z(t);z(t);x(t);q(t);p(t)]$ is uniformly bounded and the solution $x(t)$ satisfies 
\begin{equation}\label{eq:PracticalStability}
\|x(t)\|\le \epsilon, \quad \forall t\ge T_\epsilon.
\end{equation}
\hfill$\square$
\end{thm}

\begin{figure}[]
\centering
\begin{tikzpicture}[
auto,
thick,
sum/.style={circle, draw, inner sep=2pt, minimum size=2mm},
block/.style={rectangle, draw, minimum width=7.5mm, minimum height=5mm, align=center, font=\footnotesize},
block_s/.style={rectangle, draw, minimum width=7.5mm, minimum height=5mm, align=left, font=\footnotesize},
block_sat/.style={rectangle, draw, minimum width=8mm, minimum height=6mm, align=left, font=\footnotesize},
input/.style={coordinate},
output/.style={coordinate},
branch/.style={circle,inner sep=0pt,minimum size=0.5mm,fill=black,draw=black},
point/.style={circle,inner sep=0pt,minimum size=0mm,fill=black,draw=black},
skip loop/.style={to path={-- ++(0,#1) -| (\tikztotarget)}}
]

\draw
	node[input]	(input)																{}
	node[block]	(bar_G)			[right	=	4.5mm of input]						{$\bar G$}
	node[block]	(inv_bar_G)	[right	=	22mm of input]						{$\bar G^{-1}$}
	node[block_sat](sat)			[right	=	12mm of input, yshift = -7.75mm, label=left:$\Phi(w)$]{}
	node[block]	(P)				[right	=	60mm of input]						{$P$}
	node[block]	(p_dynamics)	[below	=	5mm of input, xshift=35mm]		{$\dot p = A_{a\tau}p + B_{a\tau}^p \bar G u$}
	node[block_s] (q_dynamics)	[below=	7mm of p_dynamics, xshift=27mm]		{$\begin{aligned}
\dot q &= A_{a\tau}q + B_{a \tau}^q y \\ \dot{\bar z} &= F_0(\bar z, \phi(q))\end{aligned}$}
	node[sum]		(sum0)			[right	=	15mm of input]	{}
	node[sum]		(sum1)			[below	=	1.1cm of p_dynamics]		{}
	node[output]	(output)		[right	=	14mm of P]								{}
	node[branch]	(u)				[right	=	34.625mm of input, label=above:$u$]	{}
	node[branch]	(y)				[right	=	10mm of P, label=above:$y$]			{}
	node[point]	(x0)			[left = 	-.5mm of sat]		{}
	node[point]	(x1)			[right = 	-.5mm of sat]		{}
	node[point]	(y0)			[below = 	-.5mm of sat]		{}
	node[point]	(y1)			[above = 	-.5mm of sat]		{}
	node[point]	(s1)			[below = 	1.5mm of x0, xshift = .5mm]		{}
	node[point]	(s2)			[above = 	1.5mm of x1, xshift = -.5mm]		{};
\path[font=\footnotesize]
	(input)			edge[->]	node{$u_r$}			(bar_G)
	(bar_G)			edge[->]							(sum0)
	(sum0)			edge[->]							(inv_bar_G)
	(inv_bar_G)	edge[->]							(P)
	(P)				edge[->]							(output)
	(u)				edge[->]							(p_dynamics)
	(p_dynamics)			edge[->]	node{$p_{[1]}$}	(sum1)
	(q_dynamics)			edge[->,near start]	node{$\begin{aligned}
&\;\; F(\bar z, \phi(q))\\&-[\dot q_{1\nu_1};\cdots;\dot q_{m\nu_m}]\end{aligned}$} (sum1)
	(sat)			edge[->]		(sum0)
	(x0)			edge[->]		(x1)
	(y0)			edge[->]		(y1);
\draw[->]	(y)		|-	(q_dynamics);
\draw[->,font=\footnotesize]	(sum1)	-|	node[label=30:$w$]{}	(sat);
\draw	(s1) -- +(2mm,0mm) -- ++(4.625mm,3.5625mm) -- (s2);

\end{tikzpicture}
\caption{Structure of the proposed robust controller.}\label{fig:OFDOB}
\end{figure}

\subsection{Stability analysis of the closed-loop system}
In this subsection, a proof of Theorem \ref{thm:main_result} is provided. We start by rewriting the system in new coordinates $(\xi,\eta)$ defined by  
\begin{align*}
\xi &=[\xi_1;\cdots;\xi_m] \in \R^{\nu},\quad  \xi_i = [\xi_{i1};\cdots;\xi_{i\nu_i}]\\ 
\eta &=[\eta_1;\cdots;\eta_m] \in \R^{\nu}, \quad \eta_i = [\eta_{i1};\cdots;\eta_{i\nu_i}]
\end{align*}
where $ i = 1, \dots, m$, and the $j$th components of $\xi_i$ and $\eta_i$ are given by
\begin{equation}	\label{eq:transformation}
\begin{split}
\xi_{ij} &= \frac{1}{\tau^{\nu_i-j}}(q_{ij}-x_{ij}), \ \eta_{ij} = \tau^{j-1}(p_{i1}^{(j-1)}-q_{i\nu_i}^{(j)}).
\end{split}
\end{equation}

For a given $\tau>0$, let $\Delta_i = {\rm diag}\{\tau^{\nu_i-1}, \dots, \tau, 1\}$. With $\Delta_i$, the variable $\xi_i$ is expressed compactly as  
$\xi_i = \Delta_i^{-1}(q_i- x_i)$, and $\xi$ can be written as $\xi =\Delta^{-1}(q- x)$ where  $\Delta = {\rm diag}\{\Delta_1, \dots, \Delta_m\}$. 

For convenience, we define $\mathsf a ={\rm diag}\{a_1^\top, \dots, a_m^\top\}$ and $\mathsf a_{[1]}={\rm diag}\{a_{11}, \dots, a_{m1}\}$.  
\begin{lemma}	\label{lem:lemma1}
The dynamics of $(\xi,\eta)$ is given by
\begin{equation}	\label{eq:xi_and_eta_dynamics}
\begin{split}
\tau\dot\xi &= A_\xi\xi- \tau B\Theta_\xi \\
\tau\dot\eta &= A\eta +B \Theta_\eta 
\end{split}
\end{equation}
where
\begin{align*}
&A_\xi = A_{a\tau}\big|_{\tau = 1} \\
&\Theta_\xi(z,x,\bar z, \xi, \eta,t)=  F(z,x) +G(z\! ,\! x,\! t)U_r(\bar z, \phi (x+\Delta\xi)) \\
&\ \ + G(z\! ,\! x,\! t)\bar G^{-1}(\bar z, \phi(x\!+\!\Delta\xi)) \Phi(\eta_{[1]}\!+\!F(\bar z, \phi(x+\Delta\xi)))
\\
&\Theta_\eta (z,x,\bar z, \xi, \eta,t)=  -\mathsf a\eta\\
& \ \ +\mathsf a_{[1]} \{ -F(z,x)+(\bar G(\bar z,\phi(x+\Delta\xi)) - G(z,x,t))u \}.\quad \diamond\!\!\!\!\!
\end{align*}
\end{lemma}
\begin{proof}
From the dynamics of $q$, one has $\dot q_{i\nu_i} = - \frac{a_{i1}}{\tau^{\nu_i}}(q_{i1}-y_i)= - B_i^\top B_{a_i \tau}^q (C_iq_i-y_i)$, and using this fact we can express the signal $w$ as $w= \eta_{[1]}+F(\bar z, \phi(q))$. Then, the control input $u$ becomes
\begin{equation}	\label{eq:control_input_with_xi_eta}
\begin{split}
u &= \bar G^{-1}(\bar z, \phi(x+\Delta\xi))\Phi(\eta_{[1]}+F(\bar z, \phi(x+\Delta\xi))) \\
&\quad + U_r(\bar z, \phi (x+\Delta\xi)).
\end{split}
\end{equation}
The dynamics of $\xi$ is derived as
\begin{align*}
\dot\xi & = \Delta^{-1}(\dot q - \dot x) \\
&= \Delta^{-1}\left\{ A_{a\tau} q \!+\! B_{a\tau}^qCx \!-\! Ax \!-\! B(F(z,x)\!+\!G(z,x,t)u) \right\} \\
&\begin{multlined}= \Delta^{-1}\left\{ (A-B_{a\tau}^q C)q-(A-B_{a\tau}^q C)x \right\} \\
\quad-\Delta^{-1}B( F(z,x) + G(z,x,t)u) )
\end{multlined}\\
& = \Delta^{-1}(A-B_{a\tau}^q C) \Delta\xi - \Delta^{-1}B\Theta_{\xi} \\
& = \frac{1}{\tau}A_{\xi}\xi - B\Theta_{\xi}.
\end{align*}
Meanwhile, the dynamics of $\eta$ can be obtained as follows. For $j=1,\dots, \nu_i-1$, the definition of $\eta_{ij}$ results in 
\begin{align*}
\tau\dot\eta_{ij} = \eta_{i,j+1}.
\end{align*}
To derive the dynamics of $\eta_{i\nu_i}$, we compute 
\begin{align*}
p_{i1}^{(\nu_i)}&= -\sum_{l=1}^{\nu_i}\frac{a_{i,\nu_i-l+1}}{\tau^l}p_{i1}^{(\nu_i-l)}+\frac{a_{i1}}{\tau^{\nu_i}}\bar G_i(\bar z, \phi(q))u\\
q_{i\nu_i}^{(\nu_i+1)}&= -\sum_{l=1}^{\nu_i}\frac{a_{i,\nu_i-l+1}}{\tau^l}q_{i\nu_i}^{(\nu_i-l+1)}\\
&\ \  + \frac{a_{i1}}{\tau^{\nu_i}}(F_i(z,x)+G_i(z,x,t)u).
\end{align*}
Using these relations, one has
\begin{align*}
\tau\dot\eta_{i\nu_i} &= \tau^{\nu_i }\left(p_{i1}^{(\nu_i)}-q_{i\nu_i}^{(\nu_i+1)}\right) \\
&=-a_{i1}\eta_{i1}-\cdots-a_{i\nu_i}\eta_{i\nu_i} \\
&\quad +\!a_{i1} \left\{ \bar G_i(\bar z, \phi(x\!+\!\Delta\xi))u \!-\! F_i (z,x) \! -\! G_i(z,x,t)u \right\}.
\end{align*}
Collecting the dynamics of $\eta_{ij}$ obtained above, we have the dynamics of  $\eta$ in \eqref{eq:xi_and_eta_dynamics}. Hence, the proof is complete. 
\end{proof}

With \eqref{eq:xi_and_eta_dynamics}, the closed-loop system is rewritten as
\begin{align}	\label{eq:ssp}
\dot z&=F_0(z, x)\nonumber\\
\dot{\bar z} &= F_0(\bar z, \phi(x+\Delta\xi)) \nonumber\\
\dot x &= Ax + B\Theta_\xi \nonumber\\
\tau\dot\xi &= A_\xi\xi-\tau B\Theta_\xi \\
\tau\dot\eta &= A\eta + B\Theta_\eta \nonumber
\end{align}
This equation \eqref{eq:ssp} is the standard singular perturbation form with the time separation parameter $\tau$ \cite{B_Khalil}.
If $\tau$ is sufficiently small, the {\it fast} variables $\xi$ and $\eta$ approach to their quasi-steady states $\xi^*$ and $\eta^*$ which are functions of {\it slow} variables such as $z$ and $x$.

\begin{lemma}	\label{lem:lemma2} 
Suppose that $\|z\|\le l_z$, $(\bar z,x) \in \Omega_{(\bar z, x)}(l)$. Then, the  quasi-steady states $(\xi^*,\eta^*)$ of the $(\xi,\eta)$-dynamics is given by
\begin{equation}	\label{eq:qss}
\begin{split}
\xi^* &= 0 \\
\eta^* & = C^\top \eta^*_{[1]} \\
\eta^*_{[1]}& = -F(\bar z, x) \\
&\quad + \bar G(\bar z, x)G^{-1}(z,x,t)  \big\{ F(\bar z, x) - F(z,x)\\
&\hspace{2cm}+ (\bar G(\bar z,x) - G(z,x,t))U_r(\bar z,x) \big\}.
\end{split}
\end{equation} \hfill $\diamond$
\end{lemma}
\begin{proof}
The quasi-steady states $(\xi^*, \eta^*)$ are determined from the right hand side of the  $(\xi,\eta)$-dynamics in \eqref{eq:ssp} with $\tau=0$, namely 
\begin{align*}
0 &= A_\xi\xi^* \big|_{\eta = \eta^*, \xi=\xi^*, \tau=0}\\
0 &= A\eta^* + B\Theta_\eta \big|_{\eta = \eta^*, \xi=\xi^*, \tau=0 .}
\end{align*}
Since $A_\xi$ is Hurwitz, we have $\xi^* = 0$. From the structure of $A$ and $B$, it follows that $\eta_{i2}^* = \dots = \eta_{i\nu_i}^* = 0$, $i = 1, \dots, m$. 
To find $\eta_{11}^*, \dots, \eta_{m1}^*$, we consider the dynamics of $\eta_{1\nu_1},\dots, \eta_{m\nu_m}$. Substituting the results on the components of $(\xi^*, \eta^*)$ obtained above yields 
\begin{equation}\label{eq:eta_1}
\begin{split}
0=&-{\mathsf a}_{[1]} \big[\eta_{[1]}^*+ F(z,x)  - (\bar G (\bar z,\phi(x)) - G(z,x,t) )\\
&\ \times \big\{\bar G^{-1}(\bar z, \phi(x)) \Phi (\eta_{[1]}^* + F(\bar z, \phi(x))) + U_r(\bar z,\phi(x))\big\} \big]. 
\end{split}
\end{equation}
Since the parameters $a_i$ are coefficients of stable polynomials, all of them are positive numbers. Hence, $\mathsf a_{[1]}$ is invertible. Moreover, by construction of $\Phi$ and $\phi$, the assumptions $\|z\|\le l_z$ and $(\bar z, x) \in \Omega_{\bar z,x}(l)$ ensure that $\phi(x) = x$ and $ \Phi (\eta_{[1]}^* + F(\bar z, \phi(x))) =\eta_{[1]}^* + F(\bar z, x)$. From these facts, we have
\begin{align*}
0 =&  F(z,x)\! -\! F(\bar z, x) + G(z,x,t) \bar G^{-1}(\bar z, x) (\eta_{[1]}^* + F(\bar z,x))  \\& - (\bar G (\bar z,x) - G(z,x,t) ) U_r(\bar z,x)\big\} 
\end{align*}
and one can obtain $\eta_{[1]}^*$ as in \eqref{eq:qss}. Thus, the proof is complete.
\end{proof}

With $\tilde\xi = \xi - \xi^*$ and $\tilde\eta = \eta - \eta^*$, the dynamics \eqref{eq:ssp} becomes
\begin{align*}
\dot z&= F_0(z, x)\\
\dot{\bar z} &= F_0(\bar z, \phi(x+\Delta\tilde\xi)) \\
\dot x &= Ax + B\Theta_{\tilde\xi} \\
\tau\dot{\tilde\xi} &= A_\xi\tilde\xi-\tau B\Theta_{\tilde\xi} \\
\tau\dot{\tilde\eta} &= A\tilde\eta + B\Theta_{\tilde\eta} - \tau\dot{\eta^*} 
\end{align*}
where  $\Theta_{\tilde\xi} (z,x,\bar z, \tilde \xi, \tilde \eta, t) = \Theta_{\xi}( z,x,\bar z, \tilde \xi, \tilde \eta + \eta^*, t)$ and $\Theta_{\tilde\eta}(z,x,\bar z, \tilde \xi,\tilde \eta, t) =\Theta_{ \eta}( z,x,\bar z, \tilde \xi, \tilde \eta + \eta^*, t)$. 

\begin{lemma}\label{lem:Theta_tilde_eta}
The function $\Theta_{\tilde \eta}$ can be decomposed as 
\begin{equation}\label{eq:Theta_tilde_eta}
\Theta_{\tilde \eta}(z,x,\bar z, \tilde \xi, \tilde \eta, t) =-\bar {\mathsf a}\tilde \eta  - \mathsf a_{[1]}\Psi_t(\tilde \eta_{[1]},t) + \tilde \Theta_{\tilde \eta}(\tilde\xi, t)
\end{equation}
where $\bar {\mathsf a}=\mathsf a - \mathsf a_{[1]}C$, 
\begin{equation}\label{eq:Psi_t}
\begin{split}
\Psi_t(\tilde\eta_{[1]},t) &= \tilde\eta_{[1]}-[I-G(z,x,t)\bar G^{-1}(\bar z, \phi(x ))] \\
&\quad\quad \quad \times \{\Phi(\tilde\eta_{[1]}+\eta^*_{[1]}+F(\bar z,  \phi(x ))) \\
&\quad \quad\quad\quad  \quad\quad \quad\quad-\Phi(\eta^*_{[1]}+F(\bar z,  \phi(x))) \}
\end{split}
\end{equation}
and the function $\tilde \Theta_{\tilde \eta}(\tilde \xi, t)$ is continuously differentiable with respect to its arguments and $\tilde \Theta_{\tilde \eta}(0, t)=0$. In addition, suppose $(\bar z,x)\le \Omega_{(\bar z,x)}(l)$. Then, under Assumption \ref{ass:contraction}, the function $\Psi_t$ belongs to the sector $[1-\mu, 1+\mu \big]$, i.e., $(\Psi_t(\zeta,t)\!-\!(1-\mu)\zeta)^\top (\Psi_t(\zeta,t)\!-\!(1+\mu)\zeta)\le 0$, $\forall \zeta\in \mathbb R^m$. \hfill $\diamond$
\end{lemma}
\begin{proof}
Noting that $\Theta_{\tilde \eta}(z,x,\bar z, 0,0,t)=0$, one has
\begin{align*}
\Theta_{\tilde \eta}(z,x,\bar z, \tilde \xi,\tilde \eta,t) &= \Theta_{\tilde \eta}(z,x,\bar z, \tilde \xi,\tilde \eta,t)- \Theta_{\tilde \eta}(z,x,\bar z, 0,\tilde \eta,t)\\
& +\Theta_{\tilde \eta}(z,x,\bar z, 0,\tilde \eta,t)- \Theta_{\tilde \eta}(z,x,\bar z, 0,0,t).
\end{align*}
One can easily compute that the second line of this equation becomes $-\bar {\mathsf a} \tilde \eta - a_{[1]}\Psi_t(\tilde\eta_{[1]},t)$ while the first line becomes  zero when $\tilde \xi=0$ and differentiable with respect to $\tilde \xi$ and $t$.

To prove the property of $\Psi_t$, we compute
\begin{equation*}
\begin{multlined}
\| \Psi_t(\tilde\eta_{[1]},t) - \tilde\eta_{[1]} \| \le \| [I-G(z,x,t)\bar G^{-1}(\bar z, \phi(x ))]\| \\
\times \| \Phi(\tilde\eta_{[1]}+\eta^*_{[1]}+F(\bar z,  \phi(x )))-\Phi(\eta^*_{[1]}+F(\bar z,  \phi(x))) \| 
\end{multlined}
\end{equation*}
From Assumption \ref{ass:contraction} and the facts that $\phi(x) = x$ and $\| (\partial \Phi)/(\partial \tilde\eta_{[1]}) \le 1\|$, it follows that
\begin{align*}
\| \Psi_t(\tilde\eta_{[1]},t) - \tilde\eta_{[1]} \| &\le \mu \Big\| \int_{0}^{1} \frac{\partial \Phi}{\partial \tilde\eta_{[1]}}(\sigma_v \tilde\eta_{[1]})d\sigma_v \tilde\eta_{[1]} \Big\| \\
&\le \mu \| \tilde\eta_{[1]} \|.
\end{align*}
Finally, we can obtain
\begin{align*}
&(\Psi_t(\tilde\eta_{[1]},t) \!-\! \tilde\eta_{[1]} \!+\! \mu\tilde\eta_{[1]})^\top (\Psi_t(\tilde\eta_{[1]},t) \!-\! \tilde\eta_{[1]} \!-\! \mu\tilde\eta_{[1]})\\
&=(\Psi_t(\tilde\eta_{[1]},t)-\tilde\eta_{[1]})^\top(\Psi_t(\tilde\eta_{[1]},t)-\tilde\eta_{[1]}) - \mu^2 \tilde\eta_{[1]}^\top \tilde\eta_{[1]} \le 0
\end{align*}
which implies that $\Psi_t$ belongs to the sector $(1-\mu, 1+\mu)$.
\end{proof}
 
To proceed, we define $\chi_1 = [\bar z;x]$ and $\chi_2 = [\tilde\xi;\tilde\eta]$ for the sake of simplicity.  
From the facts that the dynamics of $z, x, \bar z$ is uniformly bounded with respect to $\tau$, $\tilde \xi$, and $\tilde \eta$, and that  the boundaries of $\Omega_{\chi_1}(l)$ and $S_{\bar z}\times S_x$ are disjoint, there exists $T_0>0$, independent of $\tau$, such that $\chi_1\in \Omega_{\chi_1}(l)$ and $\|z(t)\|\le l_z$, $\forall 0\le t\le T_0$. 

Consider the Lyapunov function of $\chi_2$ defined by
\begin{align*}
V_{\chi_2} = \chi_2^{\top} P_2 \chi_2
\end{align*}
where $P_2 ={\rm diag}\{P_{\tilde\xi}, \gamma P_{\tilde\eta}\}$ with $\gamma$ being a positive constant to be chosen. The matrices $P_{\tilde\xi}$ and $P_{\tilde\eta}$ are symmetric positive definite matrices that are explained below. The matrix $P_{\tilde\xi}$ is determined from $P_{\tilde\xi} A_\xi + A_\xi^{\top} P_{\tilde\xi} = -I$, while $P_{\tilde\eta}$ is associated to the stability of $\tilde \eta$-dynamics expressed in a fast time scale $\sigma := t/\tau$, namely
\begin{equation}	\label{eq:tilde_eta_prime}
\begin{split}
\tilde\eta' :=\frac{d\tilde\eta}{d\sigma}=&( A- B \bar{\mathsf a} )\tilde\eta +B\mathsf a_{[1]}\{-\Psi_t(\tilde\eta_{[1]},t)\}\\& +B\tilde \Theta_{\tilde \eta}(\tilde \xi, t)-\tau\dot{\eta^*}.
\end{split}
\end{equation}
We assume for now that $\chi_1\in \Omega_{\chi_1}(l)$ and investigate the stability of  the system $\tilde\eta' = ( A- B \bar{\mathsf a} )\tilde\eta  +B\mathsf a_{[1]}\{-\Psi_t(\tilde\eta_{[1]},t)\}$ (the dynamics \eqref{eq:tilde_eta_prime} without $\tau\dot{\eta^*}$)  by interpreting it as a feedback system $\tilde\eta' = ( A- B \bar{\mathsf a} )\tilde\eta + B\mathsf a_{[1]}u^\dagger$ with the output $y^\dagger = \tilde\eta_{[1]}$ and the feedback $u^\dagger = -\Psi_t(y^\dagger,t)$. 
Note that the transfer function matrix from $u^\dagger$ to $y^\dagger$ is $H(s)={\rm diag}\{H_1(s),\dots, H_m(s)\}$ where $H_i(s)$ is the transfer function defined in \eqref{eq:tf}. By the way it is constructed, 
$H(s)[I_m + (1-\mu)H(s)]^{-1}$ is diagonal and Hurwitz, and $[I_m + (1+\mu)H(s)][I_m + (1-\mu)H(s)]^{-1}$ is diagonal and strictly positive real.
Since $\chi_1\in \Omega_{\chi_1}(l)$, $\Psi_t(y,t)$ belongs to the sector $[1-\mu, 1+\mu]$ by Lemma \ref{lem:Theta_tilde_eta} and the circle criterion ensures that the feedback system is exponentially stable and admits a quadratic Lyapunov function $V_{\tilde\eta} = \tilde\eta^{\top} P_{\tilde\eta} \tilde\eta$, with $P_{\tilde\eta} = P_{\tilde\eta}^\top >0$, such that $V_{\tilde \eta}'$ along the trajectory of $\tilde \eta$-dynamics (without $\tau\dot \eta^*$) satisfies $V_{\tilde\eta}' \le -\kappa_{\tilde\eta}\|\tilde\eta\|^2$ with $\kappa_{\tilde\eta}>0$.

\begin{lemma}	\label{lem:stability_chi_2}
Suppose that  $\chi_1(t) \in \Omega_{\chi_1}(l)$ and  $\|z(t)\|\le l_z$, $\forall t_0 \le t \le T$ for some $t_0$ and $T$. 
Then, there exist positive constants $\rho, \lambda$, and $\tau^*_1 \le 1$ such that for any $0 < \tau < \tau_1^*$ and $t_0 \le t \le T$, it holds that
\begin{equation*}
\dot V_{\chi_2} \le -\frac{\lambda}{2\tau}\|\chi_2\|^2, \;{\rm when}\; V_{\chi_2} \ge \rho\tau^2.
\end{equation*}
$\hfill\diamond$
\end{lemma}

\begin{proof}
The time derivative of $V_{\chi_2}$ is computed as 
\begin{align*}
\dot V_{\chi_2} &= \frac{1}{\tau}\tilde\xi^{\top}(P_{\tilde\xi} A_\xi + A_\xi^\top P_{\tilde\xi}){\tilde\xi} - 2 \tilde\xi^\top P_{\tilde\xi} B\Theta_{\tilde\xi} \\
&\quad -\frac{1}{\tau}\gamma\kappa_{\tilde\eta}\|\tilde\eta\|^2 +\frac{2\gamma}{\tau}\tilde\eta^{\top}P_{\tilde\eta}\tilde\Theta_{\tilde\eta}(\tilde \xi,t) - 2\gamma \tilde\eta^\top P_{\tilde\eta}\dot{\eta^*}.
\end{align*}
Since $\Theta_{\tilde \xi}(z,x,\bar z, \tilde \xi, \tilde \eta, t)$, $\tilde \Theta_{\tilde \eta}(\tilde \xi, t)$, $\partial \tilde\Theta_{\tilde \eta}(\tilde \xi,t)/\partial \tilde \xi$, and $\dot{\eta^*}$ are uniformly bounded when $\|z\| \le l_z$, $(\bar z, x)\in \Omega_l $, and $\tau \le 1$, it holds that 
\begin{equation}	\label{eq:dot_V_chi2_1}
\begin{split}
\dot V_{\chi_2} \le& -\frac{1}{\tau}\|\tilde\xi\|^2 -\frac{1}{\tau}\gamma \kappa_{\tilde\eta}\|\tilde\eta\|^2\\&+ \frac{2\gamma}{\tau}k_1\|\tilde \xi\|\|\tilde\eta\| + k_2 \|\tilde\xi\| + \gamma k_3 \|\tilde\eta\|
\end{split}
\end{equation}
where $k_1$, $k_2$, and $k_3$ are positive constants such that $\| \tilde\Theta_{\tilde \eta}(\tilde \xi,t)\| \le  k_1\|\tilde \xi\|$, $\|2 P_{\tilde\xi} B \Theta_{\tilde\xi}\| \le k_2$, and $\|2P_{\tilde\eta}(\dot{\eta ^*})\| \le k_3$, respectively. 

Take $\gamma <\kappa_{\tilde\eta}/(2k_1^2)$. Then, one has
\begin{align*}
\dot V_{\chi_2} \le -\frac{1}{2\tau}\|\tilde\xi\|^2 -\frac{1}{2\tau}\bar \kappa_1\|\tilde\eta\|^2+ k_2 \|\tilde\xi\| + \gamma k_3 \|\tilde\eta\|
\end{align*}
where $\bar \kappa_1=2( \gamma \kappa_{\tilde\eta} - 2\gamma^2 k_1^2)>0$. 
The time derivative \eqref{eq:dot_V_chi2_1} is rewritten as
\begin{equation}	\label{eq:dot_V_chi2}
\dot V_{\chi_2} \le -\frac{\kappa_1}{2\tau}\|\chi_2\|^2 - \|\chi_2\| \left( \frac{\kappa_1}{2\tau}\|\chi_2\| - \kappa_2 \right)
\end{equation}
where $\kappa_1 = \min\{1/2, \bar \kappa_1/2\}$ and $\kappa_2 = \sqrt{k_2^2 + \gamma^2k_3^2}$.
The proof is complete with  $\lambda = \kappa_1$ and $\rho = 4\lambda_{\max}(P_2)\frac{\kappa_2^2}{\kappa_1^2}$. 
\end{proof}

Let $\Omega_{\chi_2} = \left\{ \chi_2 \in \R^{2\nu} \big| V_{\chi_2}\le\rho\tau^2 \right\}$ which depends on $\tau$. Recall the existence of time interval $0\le t\le T_0$ during which $\chi_1(t)\in \Omega_{\chi_1}(l)$ and $\| z(t)\|\le l_z$. By Lemma \ref{lem:stability_chi_2}, with $t_0=0$ and $T=T_0$, there exists $\rho > 0$, $\tau_1^*>0$ with $\tau_1^*\le \min \{\sqrt{\lambda_{\min}(P_2)/\rho} \delta_{\tilde w}, 1\}$, and $P_2 = P_2^{\top} > 0$ such that for any $0 < \tau <\tau_1^*$, it holds that $\dot V_{\chi_2}(t) \le -\frac{\lambda_1}{\tau}V_{\chi_2}(t)$ with $\lambda_1 = \frac{\lambda}{2\lambda_{\max}(P_2)}$, provided that $V_{\chi_2}(t)\ge \rho \tau^2$ and $0\le t \le T_0$. Meanwhile, considering the compactness of the sets of initial conditions and the definition of $\tilde\xi$ and $\tilde\eta$, one can prove that $V_{\chi_2}(0) \le \frac{k}{\tau^{2r}}$ for some positive constant $k>0$ and $r = \max\{\nu_i\}$, $i=1,\dots,m$.  Applying this bound for $V_{\chi_2}(0)$, one has
$V_{\chi_2}(t) \le e^{-\lambda_1\frac{t}{\tau}}V_{\chi_2}(0)\le e^{-\lambda_1\frac{t}{\tau}}\frac{k}{\tau^{2r}}$. From the fact that $\lim_{\tau\rightarrow 0}e^{- \lambda_1 \frac{T_0}{\tau}}\frac{k}{\rho\tau^{2r+2}} =0$, one can find $\tau_2^*>0$ such that for any $\tau<\tau_2^*$, it holds that $V_{\chi_2}(T_0)\le \rho \tau^2$. Thus, for any initial condition $(p(0), q(0))\in S_p \times S_q$,  it holds that $\chi_1(t)\in \Omega_{l}$, $\forall t\le T_0$, and $\chi_2(t)$ enters the set $\Omega_{\chi_2}$ in finite time $T_0$.

We now prove that there exists $\tau_3^*>0$ such that for any $\tau<\tau_3^*$, the set $\Omega_{\chi_1}(l)\times \Omega_{\chi_2}$ is  positively invariant.
This will be carried out by showing that
\begin{equation}	\label{eq:lyapunov_property}
\begin{split}
\dot {\bar V} &\le 0, \quad \forall (\chi_1, \chi_2) \in \left\{ \chi_1 \big| \bar V(\chi_1) = l\right\} \times \Omega_{\chi_2} \\
\dot V_{\chi_2} &\le 0, \quad \forall (\chi_1, \chi_2) \in \Omega_{\chi_1}(l)\times \left\{ \chi_2 \big| V_{\chi_2} = \rho\tau^2\right\}.
\end{split}
\end{equation}
To show the first property of \eqref{eq:lyapunov_property}, we start by noting that when $\tau\le \min\{\tau_1^*,\tau_2^*\}$, $V_{\chi_1}(\bar z, x) \le l$, and $\chi_2 \in \Omega_{\chi_2}$, the saturation functions $\Phi(w)$ and $\phi(x)$ lose the effect, i.e., $\Phi(w) = w$ and $\phi(x) = x$. 
Then, the time derivative of $\bar V(\chi_1)$ becomes
\begin{align*}
\dot {\bar V}  &=  \frac{\partial \bar V}{\partial x} \left[ A x + B\left( F(z,x) + G(z,x,t) u \right) \right] \\&\quad + \frac{\partial \bar V}{\partial \bar z} F_0 (\bar z,x).
\end{align*}
By Assumption \ref{ass:nominal_control_input}, we have
\begin{multline*}
\dot {\bar V} \le -\alpha_3 \left( \| \chi_1 \| \right) + \frac{\partial \bar V}{\partial x} B[ F(z,x) - F(\bar z, x) \\ + G(z,x,t)u - \bar G(\bar z,x)U_r(\bar z, x)].
\end{multline*}
With \eqref{eq:control_input_with_xi_eta} and $\eta^*_{[1]}$ of \eqref{eq:qss}, one  has
\begin{align*}
\dot {\bar V} &\le -\alpha_3 \left( \| \chi_1 \| \right) + \frac{\partial \bar V}{\partial x}B \big[ F(z,x) - F(\bar z, x) \\
&\quad 
+ G(z,x,t) \bar G^{-1}(\bar z, x+\Delta\tilde\xi) (\tilde\eta_{[1]}\!+\!\eta_{[1]}^*\!+\!F(\bar z, x\!+\!\Delta\tilde\xi)) \\
&\quad + G(z,x,t) U_r(\bar z, x\!+\!\Delta\tilde\xi) - \bar G(\bar z,x) U_r(\bar z, x) \big]
\\
&\le  -\alpha_3 \left( \| \chi_1 \| \right) +  \frac{\partial\bar V
}{\partial x}BG(z,x,t) \bar G^{-1}(\bar z, x+\Delta \tilde\xi)\tilde\eta_{[1]} \\&\quad + \frac{\partial \bar V}{\partial x}B \Pi(\tilde\xi, t)
\end{align*}
where
\begin{align*}
&\Pi(  \tilde\xi, t)=\\& F(z,x) - F(\bar z,x) \\ 
&+ G(z,x,t)U_r(\bar z, x+ \Delta\tilde\xi) -\bar G(\bar z, x)U_r(\bar z, x) \\
& +G(z,x,t) \bar G^{-1}(\bar z,x+\Delta \tilde\xi)(F(\bar z, x+\Delta\tilde\xi) - F(\bar z, x)) \\
&+ G(z,x,t) \bar G^{-1}(\bar z,x+\Delta \tilde\xi)\bar G(\bar z,x)G^{-1}(z,x,t)\\
&\times \{ -F(z,x) + F(\bar z, x) -(G(z,x,t) - \bar G(\bar z,x))U_r(\bar z,x) \}.
\end{align*}
Note that $\Pi(\tilde\xi,t) = 0$ when $\tilde\xi = 0$.

Since $\chi_1$ and $\chi_2$ belong to compact sets, we have 
\begin{equation}	\label{eq:dot_V_bar}
\begin{split}
\dot {\bar V} &\le -\alpha_3 \left( \| \chi_1 \| \right) + \kappa_3 \left( k_{\tilde\eta} \|\tilde\eta\|+ k_{\tilde\xi} \| \tilde\xi\| \right) \\
& \le -\alpha_3(\alpha_2^{-1}(\bar V)) + \kappa   \| \chi_2 \| \\
\end{split}
\end{equation}
where $\kappa_3=\max_{\chi_1\in \Omega_{\chi_1}(l)} \left\| \frac{\partial \bar V}{\partial x}B\right\|$, $k_{\tilde\eta}$ is the maximum of $\|G(z,x,t)\bar G^{-1}(\bar z, x+\Delta\tilde\xi)\|$ on the region under consideration, $k_{\tilde\xi}$ is a constant such that $\Pi(\tilde \xi,t) \| \le k_{\tilde\xi}\|\tilde\xi\|$, and $\kappa = \kappa_3 {\rm max}\{k_{\tilde\eta}, k_{\tilde\xi} \}$. 
Let  $\bar \tau_3^* = \frac{1}{ \kappa } \sqrt{\frac{\lambda_{\min}(P_2)}{\rho}}\alpha_3(\alpha_2^{-1}(l))$ and take $\tau\le \min\{\tau_1^*, \tau_2^*,\bar \tau_3^*\}$. Since $\chi_2 \in \Omega_{\chi_2}$, it follows that  $\|\chi_2\|\le \sqrt{\frac{\rho}{\lambda_{\min}(P_2)}}\bar \tau_3^*\le\alpha_3(\alpha_2^{-1}(l))/\kappa$. 
Applying this result and the condition $\bar V(\chi_1)=l$ to \eqref{eq:dot_V_bar}, we have $\dot {\bar V}(\chi_1)\le 0$, $\forall (\chi_1,\chi_2)\in \left\{ \chi_1 \big| \bar V(\chi_1)=l \right\}\times\Omega_{\chi_2}$.

The second equation of \eqref{eq:lyapunov_property} follows from Lemma \ref{lem:stability_chi_2} with possibly new bound on $\tau$, say $\tilde \tau_3>0$. Taking $\tau_3^* =\min\{\tau_1^*, \tau_2^*, \bar \tau_3^*, \tilde \tau_3^*\}=:\tau^*$ proves the invariance. 

The proceeding arguments prove that the solution of the closed-loop system is uniformly bounded.

Now, we prove the practical stability \eqref{eq:PracticalStability}. 
Let $\epsilon>0$ be given. Since the time derivative of $\bar V(\chi_1)$ satisfies \eqref{eq:dot_V_bar}, we can achieve the objective by taking sufficiently small $\tau$ such that $\|\chi_2(t)\|$ remains as small as desired. In fact, let $\tau_4^*= \frac{1}{2\kappa}\sqrt{\lambda_{\min}(P_2)/{\rho}} \alpha_3(\alpha_2^{-1}(\alpha_1(\epsilon)))$ and take $\tau\le \min\{\tau_3^*, \tau_4^*\}$. Then by invariance, the assumption of Lemma \ref{lem:stability_chi_2} holds true with $t_0=T_0$ and $T=\infty$, and there exists $T_1>T_0$  such that $\chi_2(t)\in\Omega_{\chi_2}|_{\tau=\tau_4^*}$, $\forall t\ge T_1$ which means that $\|\chi_2(t)\| \le \alpha_3(\alpha_2^{-1}(\alpha_1(\epsilon)))/(2\kappa)$, $\forall t\ge T_1$. Applying this bound to \eqref{eq:dot_V_bar} yields
\begin{equation*}
\dot {\bar V} \le   -\alpha_3(\alpha_2^{-1}(\bar V)) + \frac{1}{2}\alpha_3  (\alpha_2^{-1}(\alpha_1(\epsilon)))
\end{equation*}
and this implies that $\dot {\bar V} \le -\frac{1}{2}\alpha_3(\alpha_2^{-1}(\bar V))$ whenever $\bar V(\chi_1)\ge \alpha_1(\epsilon)$.  Therefore, there exists $T_\epsilon>0$ such that $\chi_1(t)\le \epsilon$, $\forall t\ge T_\epsilon$. The proof is complete.

\section{Example}	\label{sec:example}
In this section, we validate the performance of the proposed strategy through an numerical example.
Consider a point mass satellite that is allowed to have two dimensional (planar) motion.
In polar coordinate, the dynamics (see Fig. \ref{fig:example}) is given as
\begin{equation}	\label{eq:example}
\begin{split}
\dot r &= v \\
\dot v &= r\omega^2 - \frac{k}{r^2}+\frac{1}{m}u_r \\
\dot\psi &= \omega \\
\dot\omega &= -\frac{2v\omega}{r} + \frac{1}{m r} u_{\psi}
\end{split}
\end{equation}
where $m$ is the mass of the satellite, $r$ is the radial distance, $\psi$ is the polar angle from the horizontal line, $u_r$ is the thrust of radial direction, $u_\theta$ is the thrust of tangential direction, and $k$ is a constant related the force field.
We assume that the satellite does not have a rotational motion by the external force, i.e., the behavior of the satellite is decided by only the control inputs $u_1$ and $u_2$, and the outputs $r$ and $\psi$ can be measured.
The relation between $(u_r, u_{\psi})$ and $(u_1, u_2)$ is described  as
\begin{equation}	\label{eq:rotation}
\begin{bmatrix}u_r \\ u_\psi \end{bmatrix} = \begin{bmatrix}\cos\theta(t) & \sin\theta(t) \\ -\sin\theta(t) & \cos\theta(t) \end{bmatrix} \begin{bmatrix}u_1 \\ u_2 \end{bmatrix} =: R(\theta(t))u.
\end{equation}
We suppose that $\theta(t) = \theta_o(t) + \tilde\theta(t)$ where $\theta_0(t)$ is known but $\tilde\theta(t)$ is unknown and satisfies $|\tilde\theta(t)| \le c_{\tilde\theta}$ where the known constant $c_{\tilde\theta}$ is given by $|\tilde\theta(t)|<\pi/5$. 
It is noted that, without a controller $(u_1 = u_2 = 0)$, the system \eqref{eq:example} admits solutions $r(t) = r_*$ and $\psi(t) = \omega_* t$ with $k = r_*^3 \omega_*^2$, where $r_* > 0$ and $\omega_*$ are constants.

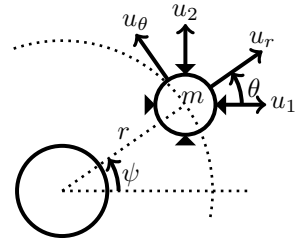
\begin{figure}[]
\centering
\begin{tikzpicture}
\draw (0,0) circle (6mm) [line width = 1.5pt];
\draw (35:2cm) circle (4mm) node[xshift = 1mm, yshift = 1mm] {$m$} [line width = 1.5pt];
\draw[dotted] (0,0) -- (35:2cm) node[midway, yshift = 1.25ex]{$r$} [line width = 1pt];
\draw[dotted] (0,0) -- (2.5cm,0) [line width = 1pt];
\draw[dotted] (-10:2cm) arc (-10:110:2) [line width = 1pt];
\fill[black] (35:2cm)++(-4mm,0) -- ++(135:2mm) -- ++(0,-2.83mm) -- cycle [line width = 1.5pt];
\fill[black] (35:2cm)++(4mm,0) -- ++(45:2mm) -- ++(0,-2.83mm) -- cycle [line width = 1.5pt];
\fill[black] (35:2cm)++(0,-4mm) -- ++(225:2mm) -- ++(2.83mm,0) -- cycle [line width = 1.5pt];
\fill[black] (35:2cm)++(0,4mm) -- ++(135:2mm) -- ++(2.83mm,0) -- cycle [line width = 1.5pt];
\draw[->] (35:2cm)++(35:4mm) -- (35:3.25cm) node[yshift = 1ex]{$u_r$} [line width = 1.5pt];
\draw[->] (35:2cm)++(125:4mm) -- +(125:0.75cm) node[yshift = 1ex]{$u_{\theta}$} [line width = 1.5pt];
\draw[->] (35:2cm)++(0,4mm) -- ++(0,6.75mm) node[yshift = 1ex]{$u_2$} [line width = 1.5pt];
\draw[->] (35:2cm)++(4mm,0) -- ++(6.75mm,0) node[xshift = 1.5ex]{$u_1$} [line width = 1.5pt];
\draw[->] (0.75,0) arc (0:35:0.75) node[midway, xshift = 1.25ex]{$\psi$} [line width = 1.5pt];
\draw[->] (35:2cm)++(7.5mm,0) arc (0:35:7.5mm) node[midway, xshift = 1.25ex]{$\theta$} [line width = 1.5pt];
\end{tikzpicture}
\caption{Example: a point mass satellite.}\label{fig:example}
\end{figure}
\begin{figure}[]
\begin{center}
\includegraphics[width=7.25cm]{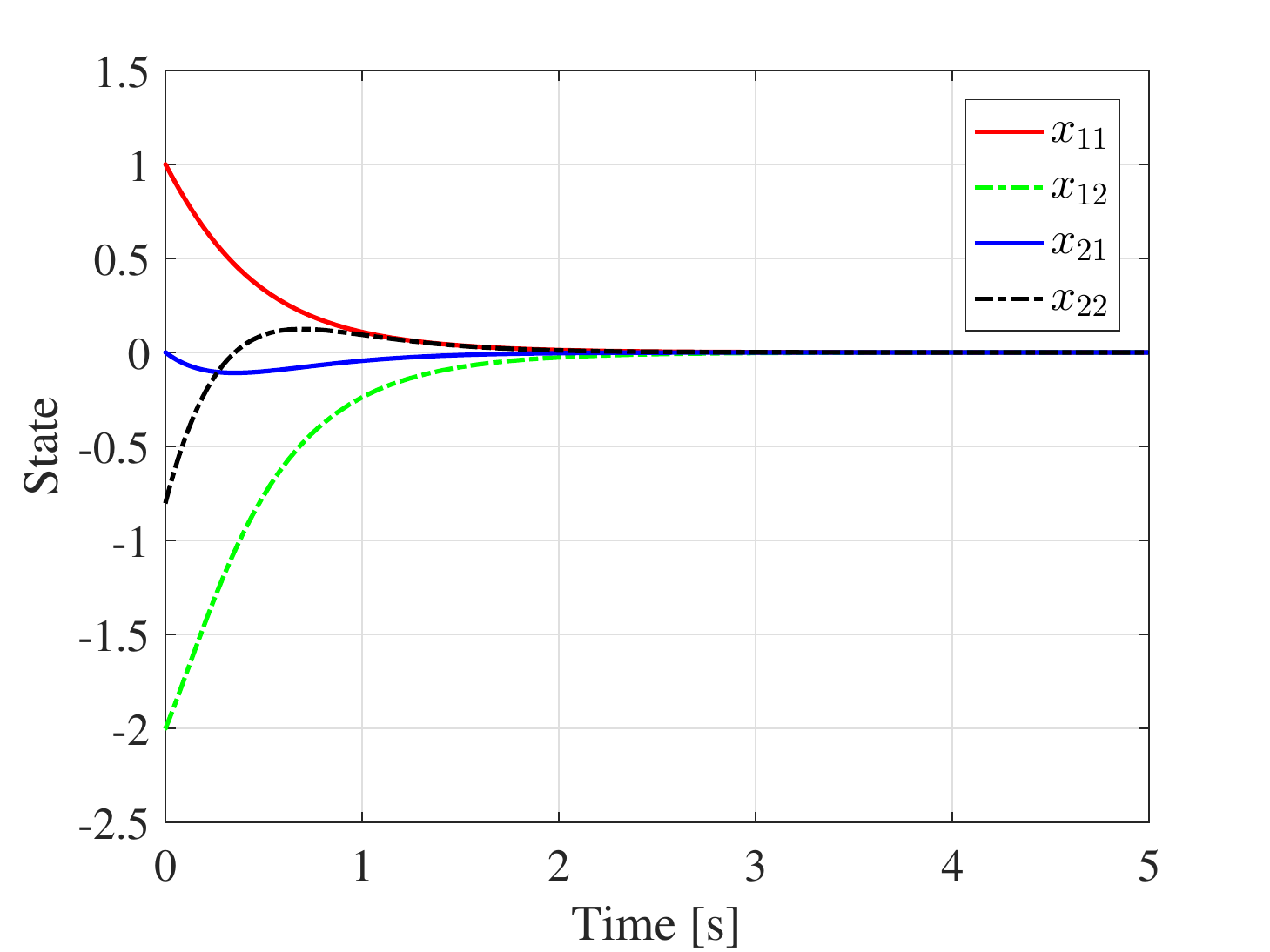}
\vspace{-2.5mm}
\caption{Performance of the nominal closed-loop system to contain the state feedback controller.}
\label{fig:Nominal}
\vspace{-1mm}
\end{center}
\end{figure}
\begin{figure}[]
\begin{center}
\includegraphics[width=7.25cm]{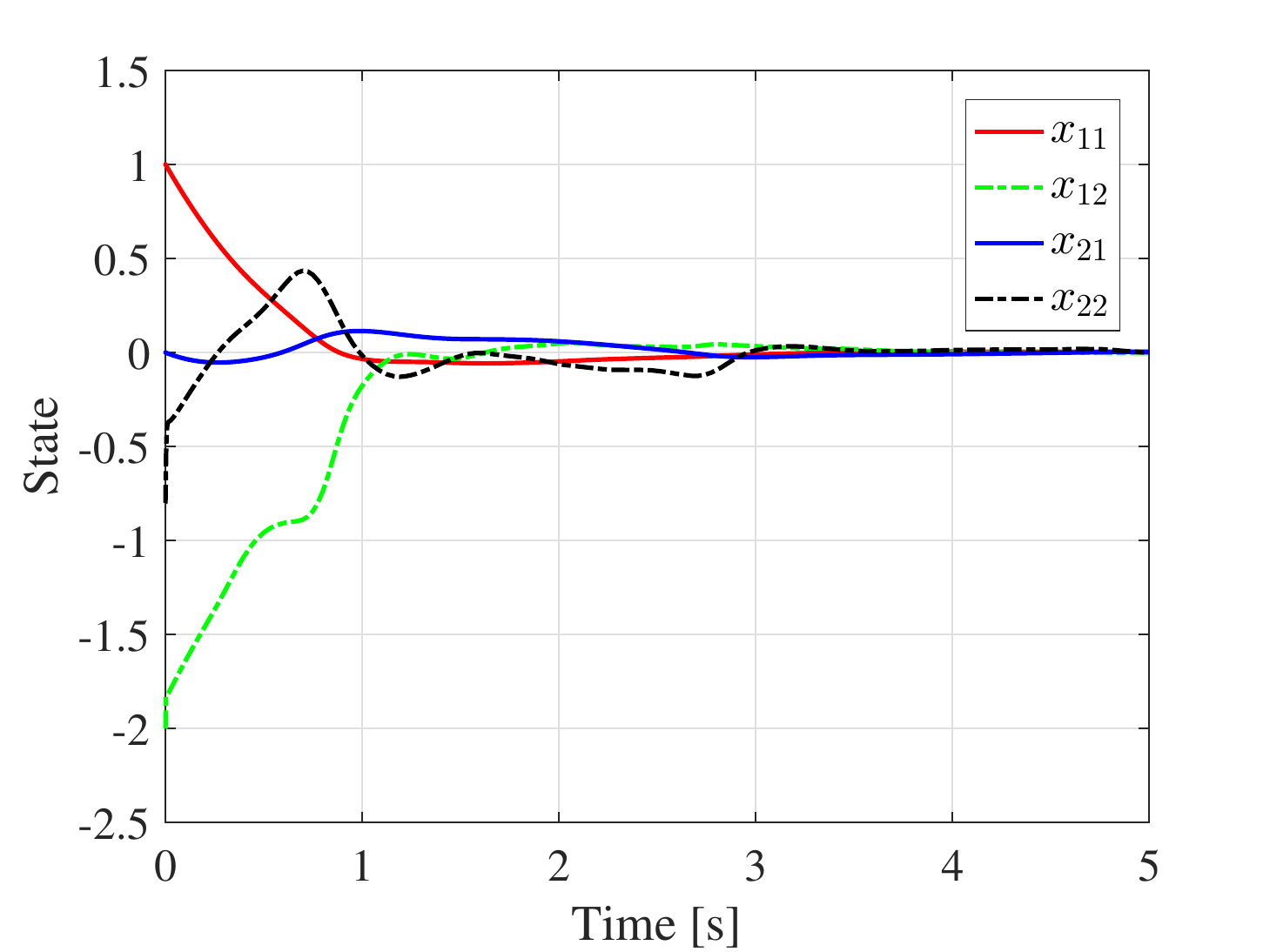}
\vspace{-2.5mm}
\caption{State trajectories of the closed-loop system under the proposed controller with a constant nominal gain matrix $\bar G_c$.}
\label{fig:Constant}
\vspace{-1mm}
\end{center}
\end{figure}
\begin{figure}[]
\begin{center}
\includegraphics[width=7.25cm]{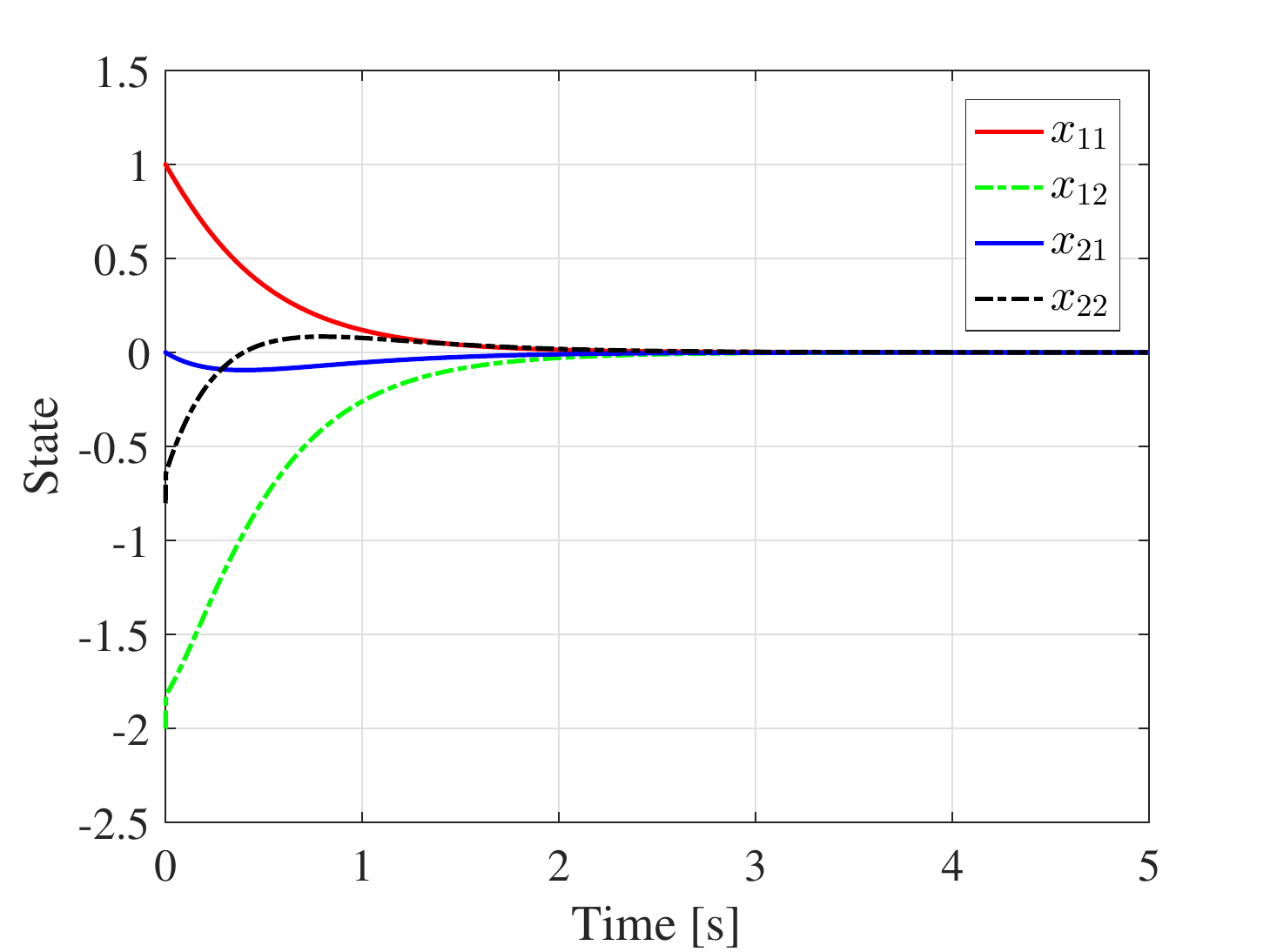}
\vspace{-2.5mm}
\caption{State trajectories of the closed-loop system under the proposed controller with a nonlinear input gain matrix  $\bar G(x,t)$.} 
\label{fig:DOB}
\vspace{-1mm}
\end{center}
\end{figure}

Define $x_{11} = r-r_*$, $x_{12} = v$, $x_{21} = r_*(\psi - \omega_* t)$, and $x_{22} = r_*(\omega - \omega_*)$.
Then, the dynamics becomes
\begin{equation}	\label{eq:transform_example}
\begin{split}
\dot x &= Ax + B(F(x) + G(x,t)u) \\
y &= Cx \\
F(x) &= \begin{bmatrix} (x_{11}+r_*)\left( \frac{x_{22}}{r_*} + \omega_* \right)^2 - \frac{k}{(x_{11}+r_*)^2} \\ -\frac{2x_{12} r_*}{x_{11}+r_*} \left( \frac{x_{22}}{r_*} + \omega_* \right) \end{bmatrix} \\
G(x,t) &= \begin{bmatrix} \frac{1}{m} & 0 \\ 0 & \frac{r_*}{m(x_{11}+r_*)}  \end{bmatrix} R(\theta(t)).
\end{split}
\end{equation}
The relative degree vector $\nu$ of the system is $[2;2]$.
To design the controller, we choose the input gain matrix $\bar G$ such that $\bar G(x,t) = {\rm diag}\{ 1/\bar m , r_*/\bar m(x_{11}+r_*) \} R(\theta_o(t))$.
Suppose the state feedback controller for the nominal system is given by 
$U_r = \bar G^{-1}(x,t)(-F(x) - K x)$.

The simulation parameters are given as follows.
$\bar m = 1$, $m = 1.2$, $k = 5$, $r_* = 1.5$, $\omega_* = \sqrt{k/r_*^3}$, $\theta_o = (\pi/2) \sin (\pi t)$, $\tilde\theta = (\pi/5) \sin (4\pi t)$, $|sI-(A-BK)|= (s+1)(s+3)^2(s+5)$, $\mu = 0.001$, $a_{11} = a_{21} = 15$, $a_{12} = a_{22} = 8$, $x(0) = [1;-2;0;-0.8]$.
The levels of saturation functions $\Phi$ and $\phi$ are taken as $100$ and $25$, respectively.

Fig. \ref{fig:Nominal} shows the state trajectories of the nominal closed-loop system under  the state feedback controller $U_r$. It is seen that the states of nominal system converge to zero.

In Fig. \ref{fig:Constant}, the state trajectories of the closed-loop system under the proposed controller with a constant nominal gain matrix $\bar G_c = {\rm diag}\{1/m, 1/m\}$ is shown and it is seen that the trajectories are distorted from the nominal ones, but the trajectory converges to the origin with very small error. 

In Fig. \ref{fig:DOB}, the proposed strategy having the nonlinear gain matrix $\bar G(x,t)$ mentioned above is applied to the actual system. One can see that, the proposed controller makes the system stable and recovers the performance of the nominal system despite the presence of system uncertainties.

Fig. \ref{fig:Input} shows the comparison of control inputs used in the simulations corresponding to Figs. \ref{fig:Constant} and \ref{fig:DOB}.
As expected, the controller with a nonlinear nominal input gain matrix requires smaller control effort than the other.

\section{Conclusions}	\label{sec:conclusion}
In this paper, we have presented a robust output feedback controller for MIMO nonlinear systems having an uncertain input gain matrix. We allow that the nominal input gain matrix can be a nonlinear function of state, which is often the case in practical systems. By proposing a new controller structure, it is shown that the performance of the state feedback controller, designed for the nonlinear nominal system, is recovered in the steady state. Compared to previous results, the proposed controller can reduce unnecessarily large control effort that comes from the difference between the actual input gain and somewhat artificial linear nominal input gain. For the future work, we will apply the proposed controller to practical systems and generalize the result to consider system uncertainties and external disturbances.
\begin{figure}[]
\begin{center}
\includegraphics[width=7.25cm]{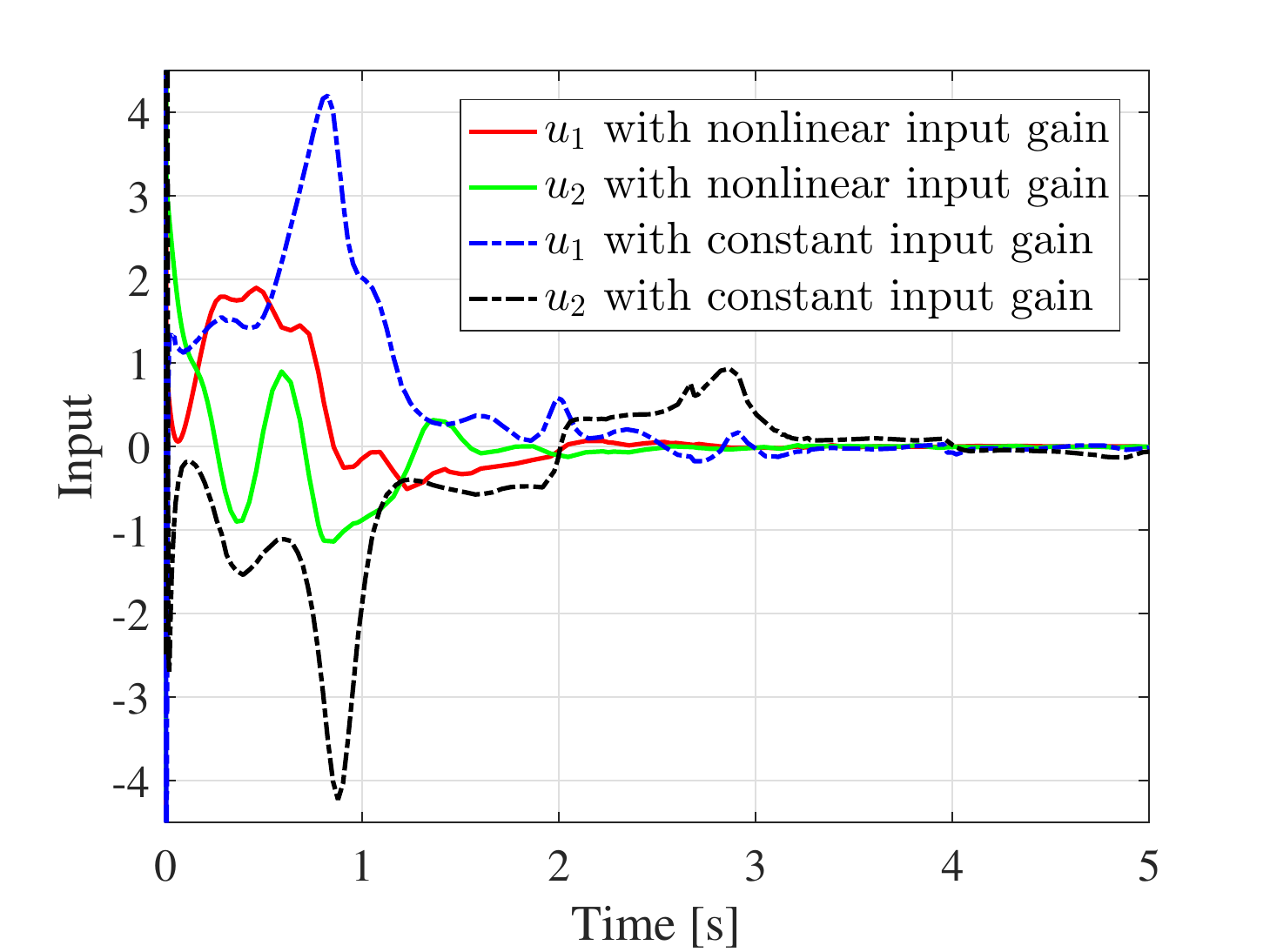}
\vspace{-2.5mm}
\caption{Comparison of control inputs using the proposed controller: constant input gain versus nonlinear input gain.} 
\label{fig:Input}
\vspace{-5mm}
\end{center}
\end{figure}
\bibliographystyle{IEEEtran}
\bibliography{References}{}             

\begin{thebibliography}{10}
\providecommand{\url}[1]{#1}
\csname url@samestyle\endcsname
\providecommand{\newblock}{\relax}
\providecommand{\bibinfo}[2]{#2}
\providecommand{\BIBentrySTDinterwordspacing}{\spaceskip=0pt\relax}
\providecommand{\BIBentryALTinterwordstretchfactor}{4}
\providecommand{\BIBentryALTinterwordspacing}{\spaceskip=\fontdimen2\font plus
\BIBentryALTinterwordstretchfactor\fontdimen3\font minus
  \fontdimen4\font\relax}
\providecommand{\BIBforeignlanguage}[2]{{%
\expandafter\ifx\csname l@#1\endcsname\relax
\typeout{** WARNING: IEEEtran.bst: No hyphenation pattern has been}%
\typeout{** loaded for the language `#1'. Using the pattern for}%
\typeout{** the default language instead.}%
\else
\language=\csname l@#1\endcsname
\fi
#2}}
\providecommand{\BIBdecl}{\relax}
\BIBdecl

\bibitem{Kailath1980}
T.~Kailath, \emph{Linear Systems}.\hskip 1em plus 0.5em minus 0.4em\relax
  Prentice-Hall, 1980.

\bibitem{B_Khalil}
H.~K. Khalil, \emph{Nonlinear Systems, {\em Third Ed}}.\hskip 1em plus 0.5em
  minus 0.4em\relax Prentice-Hall, Upper Saddle River, NJ, 2002.

\bibitem{bIsidori}
A.~Isidori, \emph{Nonlinear Control Systems, {\em Third Ed}}.\hskip 1em plus
  0.5em minus 0.4em\relax Springer-Verlag, New York, 1995.

\bibitem{Teel94}
A.~Teel and L.~Praly, ``Global stabilizability and observability imply
  semi-global stabilizability by output feedback,'' \emph{Systems \& Control
  Letters}, vol.~22, pp. 313--325, 1994.

\bibitem{AK1999}
A.~N. Atassi and H.~K. Khalil, ``A separation principle for the stabilization
  of a class of nonlinear systems,'' \emph{IEEE Trans. Automat. Contr.},
  vol.~44, no.~9, pp. 1672--1687, 1999.

\bibitem{Khalil93}
H.~K. Khalil and F.~Esfandiari, ``Semiglobal stabilization of a class of
  nonlinear systems using output feedback,'' \emph{IEEE Trans. Automat.
  Contr.}, vol.~38, no.~9, pp. 1412--1415, 1993.

\bibitem{Teel94SIAM}
A.~Teel and L.~Praly, ``Tools for semiglobal stabilization by partial state and
  output feedback,'' \emph{SIAM Journal of Control Optimal}, vol.~33, no.~5,
  pp. 1443--1488, 1994.

\bibitem{Lin1995TAC}
Z.~Lin and A.~Saberi, ``Robust semiglobal stabilization of minimum-phase
  input-output linearizable systems via partial state and output feedback,''
  \emph{IEEE Transactions on Automatic Control}, vol.~40, no.~6, pp.
  1029--1041, 1995.

\bibitem{FK2008}
L.~B. Freidovich and H.~K. Khalil, ``Performance recovery of
  feedback-linearization-based designs,'' \emph{IEEE Transactions on Automatic
  Control}, vol.~53, no.~10, pp. 2324--2334, 2008.

\bibitem{Nussbaum1983}
R.~Nussbaum, ``Some remarks on a conjecture in parameter adaptive control,''
  \emph{Systems \& Control Letters}, vol.~3, no.~5, pp. 243--246, 1983.

\bibitem{KKK95}
M.~Krsti{\'c}, I.~Kanellakopoulos, and P.~V. Kokotovi{\'c}, \emph{Nonlinear and
  Adaptive Control Design}.\hskip 1em plus 0.5em minus 0.4em\relax Wiley, New
  York, 1995.

\bibitem{Xu2003TAC}
H.~Xu and P.~A. Ioannou, ``Robust adaptive control for a class of {MIMO}
  nonlinear systems with guaranteed error bounds,'' \emph{IEEE Transactions on
  Automatic Control}, vol.~48, no.~5, pp. 728--742, 2003.

\bibitem{Ge2014TAC}
S.~S. Ge and Z.~Li, ``Robust adaptive control for a class of {MIMO} nonlinear
  systems by state and output feedback,'' \emph{IEEE Transactions on Automatic
  Control}, vol.~59, no.~6, pp. 1624--1629, 2014.

\bibitem{Wang2015}
L.~Wang, A.~Isidori, and H.~Su, ``Output feedback stabilization of nonlinear
  {MIMO} systems having uncertain high-frequency gain matrix,'' \emph{Systems
  \& Control Letters}, vol.~83, pp. 1--8, 2015.

\bibitem{BackShim2009}
J.~Back and H.~Shim, ``An inner-loop controller guaranteeing robust transient
  performance for uncertain {MIMO} nonlinear systems,'' \emph{IEEE Transactions
  on Automatic Control}, vol.~54, no.~7, pp. 1601--1607, 2009.

\bibitem{Ohnishi87}
K.~Ohnishi, ``A new servo method in mechatronics,'' \emph{Trans. of Japanese
  Society of Electrical Engineers}, vol. 107-D, pp. 83--86, 1987.

\bibitem{Chen+2016TIE}
W.~H. Chen, J.~Yang, L.~Guo, and S.~Li, ``Disturbance-observer-based control
  and related methods-an overview,'' \emph{IEEE Transactions on Industrial
  Electronics}, vol.~63, no.~2, pp. 1083--1095, 2016.

\bibitem{Shim+2016CTT}
H.~Shim, G.~Park, Y.~Joo, J.~Back, and N.~Jo, ``Yet another tutorial of
  disturbance observer: Robust stabilization and recovery of nominal
  performance,'' \emph{Control Theory and Technology}, vol.~14, no.~3, pp.
  237--249, 2016.

\bibitem{Sontag95}
E.~D. Sontag and Y.~Wang, ``On characterizations of the input-to-state
  stability property,'' \emph{Systems \& Control Letters}, vol.~5, no.~24, pp.
  351--359, 1995.

\end{thebibliography}
\end{document}